\documentclass[11pt]{article}

\usepackage[letterpaper,margin=1in]{geometry}
\usepackage[T1]{fontenc}
\usepackage{libertinus}
\usepackage{amsmath,amssymb,amsthm,mathtools}
\usepackage{aliascnt}
\usepackage{braket}
\usepackage{microtype}
\usepackage{needspace}
\usepackage{placeins}
\usepackage{enumitem}
\usepackage{booktabs}
\usepackage{array,tabularx}
\usepackage{nicematrix}
\usepackage[T1]{fontenc}

\usepackage[table]{xcolor}
\definecolor{LinkTeal}{HTML}{006F69}
\definecolor{TableHeader}{HTML}{F1F3F3}
\definecolor{TableRule}{HTML}{65716D}

\newcolumntype{C}[1]{>{\centering\arraybackslash}m{#1}}
\newlength{\comparisontextwidth}
\newcommand{\comparisoncell}[1]{\vspace{3pt}#1\par\vspace{3pt}}

\usepackage{titlesec}
\usepackage{titling}
\usepackage{hyperref}
\usepackage[nameinlink,noabbrev]{cleveref}

\allowdisplaybreaks
\numberwithin{equation}{section}
\setlist{leftmargin=1.7em,itemsep=0.15em,topsep=0.35em}

\titleformat{\section}
  {\large\bfseries}
  {\thesection}{0.65em}{}
\titleformat{\subsection}
  {\normalsize\bfseries}
  {\thesubsection}{0.6em}{}
\titlespacing*{\section}{0pt}{2.1ex plus 0.5ex minus 0.2ex}{0.8ex}
\titlespacing*{\subsection}{0pt}{1.5ex plus 0.4ex minus 0.2ex}{0.45ex}

\pretitle{\begin{center}\LARGE\bfseries}
\posttitle{\par\end{center}\vspace{0.45em}}
\preauthor{\begin{center}\normalsize}
\postauthor{\par\end{center}}
\predate{\begin{center}\small}
\postdate{\par\end{center}\vspace{-0.8em}}

\newtheorem{theorem}{Theorem}[section]
\newaliascnt{lemma}{theorem}
\newtheorem{lemma}[lemma]{Lemma}
\aliascntresetthe{lemma}
\crefname{lemma}{lemma}{lemmas}
\Crefname{lemma}{Lemma}{Lemmas}
\newaliascnt{fact}{theorem}
\newtheorem{fact}[fact]{Fact}
\aliascntresetthe{fact}
\crefname{fact}{fact}{facts}
\Crefname{fact}{Fact}{Facts}
\newaliascnt{proposition}{theorem}
\newtheorem{proposition}[proposition]{Proposition}
\aliascntresetthe{proposition}
\crefname{proposition}{proposition}{propositions}
\Crefname{proposition}{Proposition}{Propositions}
\newaliascnt{corollary}{theorem}
\newtheorem{corollary}[corollary]{Corollary}
\aliascntresetthe{corollary}
\crefname{corollary}{corollary}{corollaries}
\Crefname{corollary}{Corollary}{Corollaries}
\theoremstyle{definition}
\newaliascnt{definition}{theorem}
\newtheorem{definition}[definition]{Definition}
\aliascntresetthe{definition}
\crefname{definition}{definition}{definitions}
\Crefname{definition}{Definition}{Definitions}
\theoremstyle{remark}
\newaliascnt{remark}{theorem}
\newtheorem{remark}[remark]{Remark}
\aliascntresetthe{remark}
\crefname{remark}{remark}{remarks}
\Crefname{remark}{Remark}{Remarks}

\newcommand{\QAC}{\mathrm{QAC}^{0}}
\newcommand{\BQAC}{\mathrm{BQAC}^{0}}
\newcommand{\AC}{\mathrm{AC}^{0}}
\newcommand{\TC}{\mathrm{TC}^{0}}

\newcommand{\ConstForr}{\mathsf{ConstGap2Forr}}
\newcommand{\PromiseTC}{\operatorname{Promise}\text{-}\TC}
\newcommand{\CZ}{\mathsf{CZ}}
\newcommand{\Forr}{\operatorname{Forr}}
\newcommand{\wt}{\operatorname{wt}}
\newcommand{\polylog}{\operatorname{polylog}}

\newcommand{\GapForr}{\mathsf{Gap2Forr}}
\newcommand{\oh}{\operatorname{oh}}

\newcommand{\PromiseBQAC}{\operatorname{Promise}\text{-}\BQAC}
\newcommand{\PromiseAC}{\operatorname{Promise}\text{-}\AC}

\hypersetup{
  colorlinks=true,
  allcolors=LinkTeal,
  linkcolor=LinkTeal,
  citecolor=LinkTeal,
  urlcolor=LinkTeal,
  filecolor=LinkTeal,
  pdfborder={0 0 0},
  pdftitle={2-Fold Forrelation is in QAC0},
  pdfauthor={Francisca Vasconcelos}
}
\title{2-Fold Forrelation is in \texorpdfstring{$\QAC$}{QAC0}}
\author{Francisca Vasconcelos \vspace{0.1in}\\  UC Berkeley \vspace{0.1in}\\ \footnotesize \url{francisca@berkeley.edu}}
\date{\vspace{-0.15in}}

\begin{document}
\maketitle

\begin{abstract}
    We show that 2-fold Forrelation with inverse-polylogarithmic promise gap can be solved, with bounded error, by polynomial-size $\QAC$ circuits. Unlike the standard oracle-based Forrelation algorithm, our circuits receive the input explicitly, in the same form as the $\AC$ circuits against which Forrelation is known to be hard. At constant gap, this yields a natural promise-problem separation between $\QAC$ and $\AC$.
\end{abstract}
\vspace{.2in}
\section{Introduction}

A central open question in shallow quantum circuit complexity is whether Parity $\in \QAC$ \cite{Moore1999}. Here, $\QAC$ denotes constant-depth quantum circuits with arbitrary single-qubit gates and generalized Toffoli gates. $\QAC$ is a quantum analogue of $\AC$, or polynomial-size, constant-depth classical circuits composed of unbounded-fanin AND/OR gates and NOT gates. One fundamental difference between the two models is fan-out. $\AC$ circuits can copy a bit to arbitrarily many gates at no depth cost, whereas it remains open whether the corresponding unbounded quantum fan-out operation can be implemented by polynomial-size $\QAC$ circuits. Unbounded quantum fan-out is known to substantially increase the power of shallow quantum circuits \cite{HoyerSpalek2005} and is equivalent to Parity up to constant-depth reductions \cite{GreenHomerMoorePollett2002,HoyerSpalek2005}. Thus, determining whether Parity $\in \QAC$ is closely related to understanding whether $\QAC$ can achieve unbounded fan-out. Since Parity $\notin \AC$ \cite{Hastad1986}, a positive answer would also give an immediate decision separation between $\QAC$ and $\AC$.

Recent work has substantially advanced our understanding of $\QAC$'s powers and limitations. Upper-bound results include general Dicke-state preparation \cite{JoshiVasconcelos2026}, pseudorandom unitary constructions \cite{FoxmanParhamVasconcelosYuen2025}, exact polylogarithmic fan-out \cite{Rosenthal2021,GrierMorrisWu2026}, and the simulation of classical threshold circuits on replicated inputs \cite{GrierMorrisWu2026}. A growing body of work has developed learning algorithms and structural results for $\QAC$ \cite{BaoEscuderoGutierrez2025,DongOuYao2025,VasconcelosHuang2025,GrettaGuptaJoshi2026}, alongside increasingly strong circuit lower bounds \cite{Rosenthal2021,NadimpalliParhamVasconcelosYuen2024,AnshuDongOuYao2025,FennerGrierPadeThierauf2025,JoshiTalVasconcelosWright2025}. More recently, decision separations based on Majority have emerged as well, including a pointwise separation using replicated inputs \cite{GrierMorrisWu2026} and a correlation separation in the standard single-copy model \cite{GrettaGuptaJoshi2026}. These developments motivate the search for further natural decision problems that illuminate the power of $\QAC$ relative to $\AC$.

We show that 2-fold Forrelation provides such an example. Introduced to exhibit a large separation between quantum and randomized query complexity \cite{AaronsonAmbainis2018}, Forrelation is one of the canonical problems in quantum complexity theory and a natural candidate for studying shallow quantum--classical separations. Let $N=2^m$, and let $H_N$ denote the normalized Walsh--Hadamard matrix indexed by $u,v\in\{0,1\}^m$, with
\begin{equation}
(H_N)_{u,v}
:=
\frac{(-1)^{u\cdot v}}{\sqrt N}.
\label{eq:hadamard}
\end{equation}
For sign vectors $s,t\in\{\pm1\}^N$, their 2-fold Forrelation is
\begin{equation}
\Forr_N(s,t)
:=
\frac{1}{N}s^{\mathsf T}H_Nt
=
\frac{1}{N^{3/2}}
\sum_{u,v\in{0,1}^{m}}
s_ut_v(-1)^{u\cdot v}.
\label{eq:intro-forrelation}
\end{equation}
Our main result shows that this problem also separates shallow quantum and classical circuits.
\begin{theorem}[Informal main theorem; see \Cref{thm:main-upper}]
For every inverse-polylogarithmic gap, 2-fold Forrelation can be solved with bounded error by polynomial-size constant-depth $\QAC$ circuits, with pointwise correctness on every input satisfying the promise. In particular, at constant gap,
\begin{equation}
\ConstForr
\in
(\PromiseBQAC\cap\PromiseTC)\setminus\PromiseAC.
\label{eq:intro-separation}
\end{equation}
\end{theorem}
\noindent The classical side of the separation follows from known lower bounds for Forrelation
\cite{Tal2017,RazTal2022,BansalSinha2021}. The main technical contribution is therefore the $\QAC$ upper-bound. At first sight, such an upper-bound might seem immediate from the standard quantum algorithm for Forrelation, which already has constant query depth. That algorithm is naturally formulated using a \emph{binary-addressed} phase oracle. Writing
$s_u=(-1)^{x_u}$ and $t_v=(-1)^{y_v}$ for
$x,y\in\{0,1\}^N$, with $N=2^m$, the oracle for $x$ acts on an $m$-qubit address register as
\begin{equation}
\ket{u}
\longmapsto
(-1)^{x_u}\ket{u}.
\label{eq:binary-phase-oracle}
\end{equation}
In this binary representation, the rest of the Forrelation computation is extremely simple. The Walsh--Hadamard transform is just $H_N=H^{\otimes m}$, so the desired interference is obtained by alternating Hadamard layers with the two phase queries. The recent work of Buzet and Chailloux sharpened this observation, showing that even the computation surrounding the two queries can be restricted to two IQP circuits together with efficient classical processing \cite{BuzetChailloux2026}. Thus, from the perspective of shallow circuits, the main obstacle is not implementing the Hadamard interference pattern, but realizing the binary-addressed queries themselves from the explicit input strings.

Indeed, the oracle in \Cref{eq:binary-phase-oracle} performs a coherent indexing operation. Concretely, an $m$-qubit binary address $u$ must select one of the $N$ input wires
$x_1,\ldots,x_N$ and apply a phase determined by the selected bit. A natural way to realize this lookup is to first convert the binary address into its one-hot (or unary) representation,
\begin{equation}
\ket{u}\ket{0^N}
\longmapsto
\ket{u}\ket{\oh(u)}.
\label{eq:intro-binary-to-onehot}
\end{equation}
Once this conversion is available, the lookup becomes trivial. Namely, if $U_a$ denotes the $a$th qubit of the one-hot register, then the $N$ disjoint gates
$\CZ(X_a,U_a)$ implement
\begin{equation}
\ket{\oh(u)}
\longmapsto
(-1)^{x_u}\ket{\oh(u)},
\label{eq:intro-unary-phase-oracle}
\end{equation}
after which the binary-to-one-hot conversion can be reversed.

This exposes precisely where the standard query algorithm encounters the unresolved power of $\QAC$. In particular, unbounded quantum fan-out suffices to perform binary-to-one-hot conversion in constant depth. Namely, after distributing the $m$ address bits, all $N$ equality predicates $[u=a]$ can be evaluated in parallel. Without unbounded fan-out, however, no polynomial-size constant-depth $\QAC$ construction for this coherent indexing operation is known. Thus, directly realizing the standard binary-addressed phase oracle would require solving the binary-to-one-hot conversion problem that underlies this fan-out bottleneck.

Our construction avoids this bottleneck by working directly in the one-hot representation, rather than starting from a binary-addressed register. In particular, we exploit the fact that the $N$-qubit $W$ state
\begin{equation}
\ket{W_N}
=
\frac{1}{\sqrt N}
\sum_{u\in{0,1}^m}
\ket{\oh(u)}
\label{eq:intro-W}
\end{equation}
can be prepared in $\QAC$ \cite{GrierMorrisWu2026,JoshiVasconcelos2026}. For an input string $x\in\{0,1\}^N$, define the corresponding phase-encoded $W$ state by
\begin{equation}
\ket{W_x}
:=
\frac{1}{\sqrt N}
\sum_{u\in{0,1}^m}
(-1)^{x_u}\ket{\oh(u)}.
\label{eq:intro-phase-W}
\end{equation}
Because the address is represented one-hot, these phases can be applied transversally. In particular, the $N$ disjoint gates $\CZ(X_a,U_a)$ implement
\begin{equation}
\ket{x}\ket{W_N}
\longmapsto
\ket{x}\ket{W_x}.
\label{eq:intro-phase-W-map}
\end{equation}
Thus, the input dependence can be incorporated directly into the $W$ state, with no coherent lookup required. Crucially, the circuit coherently sums the amplitudes corresponding to all address pairs $(u,v)$, with phases $(-1)^{x_u+y_v+u\cdot v}$, so that the resulting amplitude is proportional to $\Forr_N(x,y)$ before it is squared by measurement. This gives the rare-event acceptance probability
\begin{equation}
q(x,y)
=
\frac{\Forr_N(x,y)^2}{N}.
\end{equation}
By contrast, the elementary classical two-query test accepts with probability
\begin{equation}
p(x,y)
=
\frac12+\frac{\Forr_N(x,y)}{2\sqrt N}.
\end{equation}
Thus, the classical signal appears only as a small bias around the constant baseline $1/2$, whereas quantum interference turns the same correlation into a probability with no constant background term.

Under the promise, the quantum test accepts with probability at least $\delta^2/N$ on YES instances and at most $\delta^2/(4N)$ on NO instances. Although both probabilities vanish with $N$, they differ by a constant multiplicative factor. Running $R=\Theta(N/\delta^2)$ copies in parallel and accepting if any copy accepts converts this multiplicative separation into a constant additive gap. Note that the same strategy would not work for the classical two-query test. Specifically, each trial already accepts with probability close to $1/2$. So, accepting whenever any trial succeeds would drive both cases close to acceptance probability one.

It remains to implement these repetitions in $\QAC$ using only a single copy of the input. Naïvely preparing $\ket{W_x}^{\otimes R}$ would require that each input bit $x_a$ participate in $R$ different phase gates, which in turn would require $O(R)$-fan-out of the input bits. Since unbounded fan-out is not available in $\QAC$, we instead avoid this input-reuse bottleneck using the truncated phase-state repetition construction of Grier--Morris--Wu \cite[Appendix~C.1, Claim~29]{GrierMorrisWu2026}. We briefly explain the idea, since it is the key ingredient that enables us to implement the amplification in $\QAC$. If $C_a$ denotes the number of one-hot registers selecting address $a$, then
\begin{equation}
(-1)^{\sum_{j=1}^{R}x_{u_j}}
=
(-1)^{\sum_{a=1}^{N}x_a(C_a\bmod 2)}.
\label{eq:intro-occupancy-phase}
\end{equation}
Hence all dependence on a fixed input bit $x_a$ can be consolidated into a single phase gate controlled by the parity of its occupancy $C_a$. Grier--Morris--Wu show that this parity can be computed in constant depth whenever the occupancy is below a polylogarithmic cutoff $T$. Outside this low-weight regime, the circuit may output an arbitrary fixed value. In our inverse-polylogarithmic-gap regime, the occupancies exceed this cutoff only on branches of negligible total weight. Thus, this \emph{truncated parity} computation reproduces the desired phases up to negligible error, allowing all $R$ effective tests to be realized in constant depth. Together, the one-hot interference test and the phase-state repetition construction give the $\QAC$ upper bound of \Cref{thm:main-upper}.

We next situate our result relative to the Raz--Tal oracle separation, which is also based on Forrelation. For each $N=2^m$, both results concern the same $N$-bit input strings indexed by addresses $u\in\{0,1\}^m$, but they differ in how the quantum computation accesses those bits. Raz and Tal assume coherent access through the binary-addressed phase oracle in \Cref{eq:binary-phase-oracle}, allowing their quantum algorithm to run in time polynomial in the address length $m=\log N$. Our $\QAC$ circuit instead receives the $N$ input bits directly as ordinary circuit input wires. Thus, the circuit may have size polynomial in $N$ and, hence, exponential in $m$. For this reason, our upper bound does not by itself yield the polynomial-in-$m$ algorithm required for an analogous oracle separation from $\mathsf{PH}$. The advantage of the explicit-input setting is that the $\QAC$ and $\AC$ circuits receive the input in exactly the same form. Overall, our result gives a direct decision separation between the two circuit classes on a common input model.

More broadly, our result complements the recent decision separations of Grier--Morris--Wu and Gretta--Gupta--Joshi. Grier--Morris--Wu show that $\QAC$ can simulate $\TC$ when given polynomially many copies of the input, and use this to obtain a pointwise separation from $\AC[p]$ for every fixed prime $p$ \cite{GrierMorrisWu2026}. Their separating function is a copy-augmented version of Majority. The input consists of several blocks, the function agrees with Majority when those blocks are consistent, and it takes a fixed value otherwise. Although this gives a standard total-function separation, the replicated-input structure effectively supplies at the outset the copies that unbounded fan-out would otherwise have to create within the quantum circuit. In this sense, the separation assumes access to a resource closely related to the main fan-out obstacle in distinguishing $\QAC$ from $\AC$. Gretta--Gupta--Joshi instead study Majority on its standard single-copy input and obtain $\QAC$ circuits with $1-o(1)$ correlation under the uniform distribution \cite{GrettaGuptaJoshi2026}.\footnote{Gretta--Gupta--Joshi state their separation against $\AC$. Combining their quantum upper bound with Razborov--Smolensky-type correlation bounds for Majority also yields an average-case separation from polynomial-size $\AC[p]$ circuits for every fixed prime $p$. Such circuits have only $o(1)$ correlation with Majority under the uniform distribution. See, e.g., \cite[Propositions~3.4 and~4.2]{OliveiraSanthanam2015}.} Our result likewise uses a standard single-copy input, but for the canonical Forrelation problem, and achieves bounded-error correctness on every input satisfying the promise. Thus, compared with Grier--Morris--Wu, we avoid assuming replicated access to the input, while compared with Gretta--Gupta--Joshi, we obtain a pointwise rather than distributional guarantee. The tradeoff is that Forrelation is a promise problem and our classical lower bound is only against $\AC$. \Cref{tab:decision-separations} summarizes these complementary results.

\begin{table}[t]
\centering
\begingroup
\small
\setlength{\tabcolsep}{5pt}
\setlength{\arrayrulewidth}{0.4pt}
\renewcommand{\arraystretch}{1.18}
\arrayrulecolor{TableRule}
\setlength{\comparisontextwidth}{%
\dimexpr\linewidth-10\tabcolsep-6\arrayrulewidth\relax}
\begin{NiceTabular}{|C{0.105\comparisontextwidth}|
C{0.190\comparisontextwidth}|
C{0.200\comparisontextwidth}|
C{0.245\comparisontextwidth}|
C{0.260\comparisontextwidth}|}
\CodeBefore
\rowcolor{TableHeader}{1}
\Body
\hline
\comparisoncell{\textbf{Paper}}
& \comparisoncell{\textbf{Problem}}
& \comparisoncell{\textbf{Input Model}}
& \comparisoncell{\textbf{Quantum Guarantee}}
& \comparisoncell{\textbf{Classical Hardness}} \\
\hline
\comparisoncell{\cite{GrierMorrisWu2026}}
& \comparisoncell{Copy-augmented Majority (total function)}
& \comparisoncell{Replicated input}
& \comparisoncell{Pointwise one-sided error on every input}
& \comparisoncell{Not in $\AC[p]$, for every fixed prime $p$} \\
\hline
\comparisoncell{\cite{GrettaGuptaJoshi2026}}
& \comparisoncell{Majority (total function)}
& \comparisoncell{Standard single-copy input}
& \comparisoncell{$1-o(1)$ correlation under the uniform distribution}
& \comparisoncell{Only $o(1)$ correlation for polynomial-size $\AC[p]$, for every fixed prime $p$} \\
\hline
\comparisoncell{This work}
& \comparisoncell{2-fold Forrelation (promise problem)}
& \comparisoncell{Standard single-copy input}
& \comparisoncell{Pointwise bounded error on every promised input}
& \comparisoncell{ Depth-$d$ randomized $\AC$ requires $\exp(N^{\Omega_d(1)})$ size} \\
\hline
\end{NiceTabular}

\smallskip
\endgroup
\caption{Comparison of recent decision separations for shallow quantum circuits.
The results differ in their input model, quantum guarantee, and strength of
the classical lower bound.}
\label{tab:decision-separations}
\end{table}

The remainder of the paper is organized as follows. \Cref{sec:preliminaries} fixes the circuit model and Forrelation promise and records the prior results used in the proof. \Cref{sec:main-results} states the formal quantum upper bound and the resulting explicit-input separation. Finally, \Cref{sec:upper-bound} proves the quantum upper-bound by constructing a single one-hot Forrelation test and amplifying it in constant depth using the phase-state repetition construction of Grier--Morris--Wu.

\FloatBarrier
\section{Preliminaries}
\label{sec:preliminaries}

We fix the circuit conventions and Forrelation promise used throughout the paper, recording the quantum primitives and classical bounds needed in the proof.

\subsection{Circuit Classes and Forrelation}
\label{sec:model-forrelation}
A $\QAC$ circuit has constant depth and polynomial size, consisting of arbitrary one-qubit gates and arbitrary-width generalized Toffoli gates (equivalently, generalized controlled-$Z$ gates) with disjoint supports within each layer. All ancillae are initialized to $\ket0$. Unbounded fan-out, intermediate measurements, postselection, and oracle gates are not available. We write $\PromiseBQAC$ for the corresponding bounded-error promise class, with completeness $2/3$ and soundness $1/3$. Classically, $\AC$ denotes polynomial-size constant-depth circuits with unbounded-fanin AND/OR gates, NOT gates, and unrestricted fan-out. $\TC$ additionally allows threshold gates. We write $\PromiseAC$ and $\PromiseTC$ for the corresponding promise classes.

Throughout, $N=2^m$, so each Forrelation input string has length $N$ and its coordinates are indexed by $\{0,1\}^m$. The total input length is therefore $2N$, while $m=\log_2N$ is the address length. We use the definition of $\Forr_N(x,y)$ from \Cref{eq:intro-forrelation}.
\begin{definition}[Gap 2-fold Forrelation]
\label{def:gap-forr}
For $\delta=\delta(N)\in(0,1]$, define
\begin{equation}
\GapForr_{N,\delta}(x,y)=
\begin{cases}
1,&\Forr_N(x,y)\ge\delta,\\
0,&|\Forr_N(x,y)|\le\delta/2,
\end{cases}
\label{eq:gap-promise}
\end{equation}
with no requirement outside the promise.
\end{definition}
\noindent For the constant-gap separation, we fix $\delta_0:=2^{-10}$ and write
$$
    \ConstForr
    :=
    \{\GapForr_{N,\delta_0}\}_{N=2^m}.
$$

\subsection{$\QAC$ primitives}
\label{sec:primitives}

The proof uses four ingredients from prior work: 1) exact polylogarithmic fan-out, 2) exact preparation of $W$ states, 3) low-weight Boolean computations, and 4) a phase-state repetition construction. We record the precise forms needed below.

\paragraph{1) Exact Polylogarithmic Fan-Out.}
We begin with the limited form of fan-out that will be used throughout the construction. Rosenthal showed that Parity and quantum fan-out can be approximated in constant depth by $\QAC$ circuits with size exponential in the fan-out width \cite[Corollary~1.2]{Rosenthal2021}. Restricting the width to $\polylog N$ therefore already gives polynomial-size approximate fan-out, an observation used, for example, in the constant-depth pseudorandomness constructions of \cite{FoxmanParhamVasconcelosYuen2025}. More recently, Grier--Morris--Wu obtained an exact polynomial-size constant-depth implementation for polylogarithmic fan-out \cite[Corollary~10]{GrierMorrisWu2026}.

\begin{fact}[Exact polylogarithmic fan-out]
\label{fact:small-fanout}
For every fixed $a\ge1$, fan-out to at most $\log^aN$ targets has an exact polynomial-size constant-depth $\QAC$ implementation.
\end{fact}

\noindent This limited form of fan-out does not resolve whether Parity lies in $\QAC$, or equivalently whether unrestricted fan-out admits a polynomial-size constant-depth $\QAC$ implementation. It nevertheless suffices for the parallel label-extraction and low-weight computations needed in our Forrelation implementation.

\paragraph{2) Exact Preparation of $W$ States.}
We next use the fact that the $N$-qubit $W$ state can be prepared exactly by polynomial-size constant-depth $\QAC$ circuits \cite[Theorem~17]{GrierMorrisWu2026}\cite[Theorem~1]{JoshiVasconcelos2026}.

\begin{fact}[Exact $W$-state preparation]
\label{fact:W-preparation}
There exists a clean polynomial-size constant-depth $\QAC$ unitary $U_W$ and $A=\operatorname{poly}(N)$ such that
\begin{equation}
U_W\ket{0^N}\ket{0^A}
=
\ket{W_N}\ket{0^A}.
\label{eq:exact-W}
\end{equation}
\end{fact}

\noindent Together with the transversal phase gates described in the introduction, \Cref{fact:W-preparation} gives an exact clean preparation of the phase-encoded state $\ket{W_x}$ while retaining the input register.

\paragraph{3) Low-Weight Boolean Computations.} The repetition argument additionally uses Boolean computations that are efficient only in the low-Hamming-weight regime. Grier--Morris--Wu show that, on an $L$-bit register, the threshold predicate $\mathrm{TH}_{\ge k}$ has an exact polynomial-size constant-depth $\QAC$ implementation whenever $k\le\log^aL$ for a fixed $a$ \cite[Lemma~12]{GrierMorrisWu2026}.

\begin{fact}[Low-weight thresholds]
\label{fact:low-weight-thresholds}
For every fixed $a\ge1$ and $k\le\log^aL$, the predicate $\mathrm{TH}_{\ge k}$ has an exact polynomial-size constant-depth $\QAC$ implementation.
\end{fact}

\noindent We use these Boolean computations coherently, computing their outputs into clean ancillae and then uncomputing the workspace. For a cutoff $T$ and input $z\in\{0,1\}^{L}$, define the truncated parity function
\begin{equation}
h_T(z):=
\begin{cases}
\wt(z)\bmod2,&\wt(z)\le T\\
0,&\wt(z)>T
\end{cases}.
\label{eq:truncated-parity}
\end{equation}
Thus $h_T$ agrees with Parity on all strings of Hamming weight at most $T$, while making no attempt to compute Parity beyond that regime. When $T=\polylog N$, $h_T$ is a symmetric $\AC$ function. In particular, it depends only on Hamming weight and is identically zero above the polylogarithmic cutoff $T$. By the exact simulation of symmetric $\AC$ functions in $\QAC$ due to Grier--Morris--Wu \cite[Corollary~16]{GrierMorrisWu2026}, $h_T$ therefore has an exact polynomial-size constant-depth $\QAC$ implementation. This truncated-parity computation is the ingredient used in their phase-state repetition construction \cite[Appendix~C.1]{GrierMorrisWu2026}.

\begin{fact}[Truncated parity]
\label{fact:truncated-parity}
For every fixed $a\ge1$ and $T\le\log^aL$, the function $h_T$ has an exact clean polynomial-size constant-depth $\QAC$ implementation.
\end{fact}

\noindent The low-weight restriction is essential. \Cref{fact:truncated-parity} does not imply a constant-depth $\QAC$ circuit for unrestricted Parity.

\paragraph{4) Phase-State Repetition Construction.} Grier--Morris--Wu combine these ingredients to prepare many effective copies of a phase-encoded $W$ state while using each input bit only once \cite[Appendix~C.1, Claim~29]{GrierMorrisWu2026}. The idea is to begin with $R$ copies of $\ket{W_N}$ and group the desired input-dependent phases according to address occupancy. If $C_a$ denotes the number of copies occupying address $a$, then all occurrences of the input bit $x_a$ contribute the single phase $(-1)^{x_a(C_a\bmod 2)}$. Thus one can compute the parity of each occupancy, apply a single phase gate involving $x_a$, and then uncompute. Replacing the occupancy parity by the truncated-parity computation above gives the following specialization of their result.

\begin{proposition}[Grier--Morris--Wu phase-state repetition, specialized]
\label{prop:phase-repetition}
Fix $b\ge0$. For $N\le R\le N\log^bN$, there is a polynomial-size constant-depth $\QAC$ circuit that, on input $\ket{x}$ and clean ancillas, prepares
$\ket{x}\ket{\rho_x}$ satisfying
\begin{equation}
\bigl|\ket{\rho_x}-\ket{W_x}^{\otimes R}\bigr|_2
\le
\exp[-\Omega(\log^2N\log\log N)]
=
N^{-\omega(1)}.
\label{eq:phase-repetition-error}
\end{equation}
The input is preserved and all ancillary work registers other than the $R$ one-hot registers return to zero.
\end{proposition}

\noindent The approximation arises only from truncating the occupancy parities. When $R\le N\log^bN$, a balls-into-bins bound shows that any occupancy exceeds the required polylogarithmic cutoff only on branches of total squared amplitude $\exp[-\Omega(\log^2N\log\log N)].$
On every other branch, truncated parity agrees with the true occupancy parity, which gives the approximation guarantee in \Cref{prop:phase-repetition} uniformly in $x$.

\subsection{Forrelation Classical Circuit Lower Bounds}
\label{sec:classical-bounds}

The classical side of our separation follows from existing lower bounds for
constant-gap Forrelation. Recall that $\ConstForr$ distinguishes inputs with
$\Forr_N(x,y)\ge\delta_0$ from those with
$|\Forr_N(x,y)|\le\delta_0/2$, for the fixed constant
$\delta_0=2^{-10}$.

Bansal--Sinha analyze two distributions on the $2N$ input bits
\cite[Theorems~3.1, 3.2, and~3.4]{BansalSinha2021}. The first, $p_0$, is
uniform and produces a NO instance with probability $1-O(1/N)$, while the
second, $p_1$, is a correlated Forrelation distribution that produces a YES
instance with constant probability. For 2-fold Forrelation, their
Fourier-analytic bound controls the distinguishing advantage between $p_1$
and $p_0$ by $N^{-1/2}$ times the level-two Fourier weight of the
distinguisher under biased product measures. Their restriction argument
\cite[Theorem~3.4]{BansalSinha2021} reduces the required biased-measure
bound to a uniform-measure Fourier bound for restrictions of the
distinguisher. Tal's Fourier-growth theorem gives precisely such a bound. Concretely,
for a depth-$d$, size-$S$ $\AC$ circuit, the level-$\ell$ Fourier
$L_1$-weight is at most $O\!\left(\log^{d-1} S\right)^\ell$ \cite[Theorem~2]{Tal2017}. Since restrictions of an $\AC$ circuit remain
$\AC$ circuits of no greater size or depth, taking $\ell=2$ yields
distinguishing advantage at most
\[
    O_d\!\left(
        \frac{(\log S)^{2(d-1)}}{\sqrt N}
    \right).
\]
Thus, distinguishing $p_1$ from $p_0$ with constant advantage requires
$S=\exp(N^{\Omega_d(1)})$. This constant-gap randomized-$\AC$ lower bound
is also summarized by Girish--Servedio
\cite[Section~1.2, Table~1]{GirishServedio2026}.

To pass from this distributional statement to a lower bound for
bounded-error promise circuits, first amplify any putative solver to
sufficiently small constant error using a constant number of independent
repetitions. This preserves constant depth and increases the size by only
a constant factor. Since $p_1$ places constant probability on YES
instances while $p_0$ produces a NO instance with probability
$1-O(1/N)$, the amplified solver would distinguish $p_1$ from $p_0$ with
constant advantage.

\begin{fact}[Constant-gap Forrelation hardness]
\label{fact:classical-forrelation-hardness}
For every fixed depth $d$, any bounded-error randomized $\AC$ circuit
solving $\ConstForr$ has size $\exp(N^{\Omega_d(1)})$.
In particular, quasipolynomial-size constant-depth $\AC$ circuits do not
suffice.
\end{fact}

\noindent
Note that the distributions are used only to establish the lower bound. The
conclusion itself is pointwise and rules out bounded-error circuits that
are correct on every input satisfying the promise.
\section{Main Results}
\label{sec:main-results}

We first state our quantum upper bound formally.

\begin{theorem}[Quantum upper bound]
\label{thm:main-upper}
Fix $c\ge0$. If $\delta(N)\ge\log^{-c}N$ for all sufficiently large
$N$, then $\GapForr_{N,\delta}\in\PromiseBQAC$. The circuit receives
the explicit $2N$-bit input $(x,y)$ and has depth $O_c(1)$, size
$N^{O_c(1)}$, and polynomial total width in $N$. 
\end{theorem}

\noindent For comparison, the same explicit-input promise is easy for
unrestricted threshold circuits.

\begin{proposition}[Elementary threshold upper bound]
\label{prop:TC-upper}
For every gap function $\delta(N)\in(0,1]$, the promise
$\GapForr_{N,\delta}$ has deterministic constant-depth threshold
circuits of size $O(N^2)$. In particular,

$$
    \GapForr_{N,\delta}\in\PromiseTC.
$$

\end{proposition}

\begin{proof}
Recall that the coordinates of $x$ and $y$ are indexed by
$u,v\in\{0,1\}^m$. For each pair of addresses $(u,v)$, the
corresponding term in the Forrelation sum is
$(-1)^{x_u+y_v+u\cdot v}$. Let
$I_{u,v}:=\mathbf1[x_u\oplus y_v=u\cdot v]$ indicate that this term is
positive, and let $B:=\sum_{u,v}I_{u,v}$ be the number of positive
terms. Since the remaining $N^2-B$ terms are negative,
\begin{equation}
\Forr_N(x,y)
=
\frac{B-(N^2-B)}{N^{3/2}}
=
\frac{2B-N^2}{N^{3/2}}.
\end{equation}
Each $I_{u,v}$ depends only on $x_u$, $y_v$, and the hardwired value
of $u\cdot v$, so all $N^2$ indicators can be computed in parallel by
constant-size gates. A single threshold gate then distinguishes the
two promised ranges.
\end{proof}

\noindent Combining \Cref{thm:main-upper} and \Cref{prop:TC-upper} with the known classical
lower bound in \Cref{fact:classical-forrelation-hardness} yields the
following separation.

\begin{corollary}[Explicit-input separation]
\label{cor:separation}
The constant-gap promise family $\ConstForr$ satisfies
\begin{equation}
\ConstForr
\in
(\PromiseBQAC\cap\PromiseTC)\setminus\PromiseAC.
\label{eq:separation}
\end{equation}
Moreover, for every fixed depth $d$, bounded-error randomized $\AC$
circuits solving $\ConstForr$ require size
$\exp(N^{\Omega_d(1)})$.
\end{corollary}

\begin{proof}
Apply \Cref{thm:main-upper} at $\delta=\delta_0$. The $\PromiseTC$
containment follows from \Cref{prop:TC-upper}, while the classical
exclusion and size lower bound are given by
\Cref{fact:classical-forrelation-hardness}. All three statements use
the same explicit input of length $2N$.
\end{proof}

\section{Proof of the $\QAC$ Forrelation circuit}
\label{sec:upper-bound}

We now prove \Cref{thm:main-upper}. The construction has two main steps. First, starting from the phase-encoded $W$ states $\ket{W_x}$ and $\ket{W_y}$, we reproduce the Walsh--Hadamard interference pattern directly on their one-hot address registers. This gives a single test with acceptance probability $\Forr_N(x,y)^2/N$. We then use the phase-state repetition construction of \Cref{prop:phase-repetition} to realize sufficiently many copies of this test in constant depth from a single copy of each input string.

\subsection{A Single One-Hot Forrelation Test}
\label{sec:one-hot-test}

The phase-encoded states $\ket{W_x}$ and $\ket{W_y}$ already supply the input-dependent signs $(-1)^{x_u+y_v}$ appearing in the Forrelation sum. It remains to introduce the Walsh--Hadamard sign $(-1)^{u\cdot v}$. Earlier, we discussed the difficulty of converting a binary address into its one-hot representation. Here we need a simpler operation in the reverse direction. Starting from a one-hot address $\ket{\oh(u)}$, we coherently extract its $m$-bit binary label $u$. Importantly, we only append the binary label and do not erase the one-hot register.

\begin{lemma}[One-hot-to-binary label extraction]
\label{lem:extractor}
There is an exact clean polynomial-size constant-depth $\QAC$ unitary
$E_N$ such that, for every $u,z\in\{0,1\}^m$,
\begin{equation}
E_N:
\quad
\ket{\oh(u)}_U\ket{z}_B
\longmapsto
\ket{\oh(u)}_U\ket{z\oplus u}_B.
\label{eq:extractor}
\end{equation}
\end{lemma}

\begin{proof}
For each address bit $r\in[m]$, let
$S_r:=\{v\in\{0,1\}^m:v_r=1\}$ be the set of addresses whose $r$th bit is one. Since $\ket{\oh(u)}$ has its unique nonzero entry at the coordinate indexed by $u$,
\begin{equation}
u_r
=
\bigvee_{v\in S_r}\oh(u)_v.
\label{eq:label-bit-OR}
\end{equation}
To compute all $m$ bits of $u$ in parallel, each one-hot qubit may need to participate in up to $m=\log_2N$ of these OR computations. By \Cref{fact:small-fanout}, exact fan-out to this many targets is available in polynomial-size constant depth. Use it to distribute each one-hot qubit to the required work registers. For each $r$, reversibly compute the OR in \Cref{eq:label-bit-OR} into the label qubit $B_r$. The $m$ OR computations act on disjoint copies and therefore run in parallel. Finally, reverse the fan-out to return all copied ancillae to zero.
\end{proof}

The distinction from the binary-to-one-hot conversion discussed earlier is important. \Cref{lem:extractor} does not implement the clean conversion
$\ket{\oh(u)}\ket{0^m}\mapsto\ket{0^N}\ket{u}$. Rather, it only appends the binary label while retaining the one-hot register. For our purposes, however, these labels are enough to implement the remaining Forrelation phase.

\begin{lemma}[One-hot Walsh--Hadamard kernel]
\label{lem:kernel}
On two one-hot registers, the unitary
\begin{equation}
K_N\ket{\oh(u),\oh(v)}
=
(-1)^{u\cdot v}\ket{\oh(u),\oh(v)}
\label{eq:kernel}
\end{equation}
has an exact clean polynomial-size constant-depth $\QAC$ implementation.
\end{lemma}
\begin{proof}
First apply $E_N$ separately to the two one-hot registers to append the binary labels $u$ and $v$. Next, apply controlled-$Z$ gates between the corresponding label qubits. Since the $r$th gate contributes the phase $(-1)^{u_rv_r}$, the product of the $m$ gates contributes $\prod_{r=1}^m (-1)^{u_rv_r}=(-1)^{u\cdot v}$, as desired.
Finally, apply $E_N^\dagger$ to both registers to uncompute the binary labels, leaving only the desired phase on the original one-hot registers.
\end{proof}

\noindent The kernel now supplies the remaining signs in the Forrelation sum. Projecting both one-hot registers back onto $\ket{W_N}$ causes the signed amplitudes to interfere and yields the desired test.

\begin{lemma}[Single-pair Forrelation test]
\label{lem:single-test}
There is a polynomial-size constant-depth $\QAC$ test on
$\ket{W_x}\ket{W_y}$ whose acceptance probability is exactly
\begin{equation}
q(x,y)
=
\frac{\Forr_N(x,y)^2}{N}.
\label{eq:single-test-probability}
\end{equation}
\end{lemma}

\begin{proof}
After applying the kernel,
\begin{equation}
K_N\ket{W_x}\ket{W_y}
=
\frac1N
\sum_{u,v}
(-1)^{x_u+y_v+u\cdot v}
\ket{\oh(u),\oh(v)}.
\end{equation}
Its overlap with $\ket{W_N}\ket{W_N}$ is
\begin{align}
\alpha(x,y)
&:=
(\bra{W_N}\otimes\bra{W_N})
K_N\ket{W_x}\ket{W_y}
=
\frac1{N^2}
\sum_{u,v}
(-1)^{x_u+y_v+u\cdot v}
=
\frac{\Forr_N(x,y)}{\sqrt N}.
\label{eq:single-amplitude}
\end{align}
Apply $U_W^\dagger$ separately to the two one-hot registers, together with their clean preparation ancillae. A generalized NOR toggles an output flag precisely on the all-zero component. Measuring this flag therefore accepts with probability
\begin{equation}
    |\alpha(x,y)|^2
    =
    \frac{\Forr_N(x,y)^2}{N}.
\end{equation}
\end{proof}

\begin{remark}[Loss of the sign]
The test depends on $\Forr_N(x,y)$ only through its square. It therefore does not distinguish positive from negative Forrelation, which is why our promise separates a positive YES region from a small-magnitude NO region.
\end{remark}

\subsection{Amplification and Proof of \Cref{thm:main-upper}}
\label{sec:amplification}

By \Cref{lem:single-test}, a single Forrelation test has acceptance
probability $\Forr_N(x,y)^2/N$. Under the promise in
\Cref{def:gap-forr}, this gives a constant multiplicative gap between the
YES and NO cases. We amplify this gap coherently by preparing
$R=\Theta(N/\delta^2)$ effective copies of the test, running them in
parallel into separate flag qubits, reversibly computing the OR of these
flags into a single output qubit, and measuring only this final output.
Since $\delta\ge\log^{-c}N$, we have $R\le N\polylog N$, so
\Cref{prop:phase-repetition} lets us realize these effective copies from a
single copy of each input string with negligible error.

\begin{proof}[Proof of \Cref{thm:main-upper}]
Choose
\begin{equation}
R
:=
\left\lceil
\frac{5N}{4\delta^2}
\right\rceil.
\label{eq:R-choice}
\end{equation}
For sufficiently large $N$, we have
$N\le R\le N\log^{2c+1}N$, so \Cref{prop:phase-repetition} applies.
Applying it separately to $x$ and $y$ gives a joint state within
$2\eta_N$ of $\ket{W_x}^{\otimes R}\otimes\ket{W_y}^{\otimes R}$,
in trace distance, where $\eta_N=N^{-\omega(1)}$.

Pair the $R$ phase-encoded registers for $x$ with those for $y$, and
coherently run the test from \Cref{lem:single-test} on every pair in
parallel, writing the outcome of each test into a flag qubit. Reversibly
compute the OR of the $R$ flags into a fresh output qubit and measure only
this output. For the ideal product state, let $P_R(x,y)$ denote the probability that
the final output qubit is $1$. By \Cref{lem:single-test},
\begin{equation}
P_R(x,y)
=
1-
\left(
1-\frac{\Forr_N(x,y)^2}{N}
\right)^R,
\label{eq:parallel-acceptance}
\end{equation}
since the output is $0$ exactly when all $R$ coherent test flags are $0$. On a YES input,
\begin{equation}
P_R(x,y)
\ge
1-\exp(-R\delta^2/N)
\ge
1-e^{-5/4}
>
0.71.
\label{eq:yes-bound}
\end{equation}
On a NO input,
\begin{equation}
P_R(x,y)
\le
\frac{R\delta^2}{4N}
\le
\frac5{16}
+
\frac{\delta^2}{4N}
<
\frac13
\end{equation}
for sufficiently large $N$, using $1-(1-q)^R\le Rq$. The $2\eta_N$ trace-distance error changes the final acceptance probability
by at most $2\eta_N$. Since $\eta_N=N^{-\omega(1)}$, the actual circuit
therefore has completeness at least $2/3$ and soundness at most $1/3$ for
sufficiently large $N$.

Every stage has depth $O_c(1)$ and polynomial size in $N$. The two banks of one-hot registers use $2RN=\Theta(N^2/\delta^2)$ qubits, with only polynomial additional workspace for state preparation, label extraction, and truncated parity. Hence the overall circuit has constant depth and polynomial size in the input length $N$. Equivalently, since $N=2^m$, its size may be exponential in the address length $m$. This proves \Cref{thm:main-upper}.
\end{proof}
\section{Acknowledgments and AI Statement}
The author is supported by NSF grant 2311733 and DOE grant DE-SC0024124.
The author thanks Uma Girish, Malvika Joshi, Jackson Morris, Gregory Rosenthal, Avishay Tal, and John Wright for insightful discussions and feedback on the manuscript. In particular, Avishay Tal suggested Forrelation as a natural candidate for a decision separation between $\QAC$ and $\AC$.

The author also acknowledges the use of AI in developing and writing this manuscript. Specifically, the author had worked out the main one-hot Forrelation construction, but initially believed that the amplification step was obstructed by the need to reuse the input bits, closely related to the binary-to-one-hot and fan-out bottlenecks discussed in the paper. After a conversation with Jackson Morris about the techniques in \cite{GrierMorrisWu2026}, the author used ChatGPT (GPT-5.6 Sol) to examine that work for potentially applicable constructions. It identified the phase-state repetition construction of Grier--Morris--Wu as a possible way to amplify the Forrelation signal without replicating the input. The author subsequently verified and adapted this construction to the proof presented in this manuscript. ChatGPT was also used, under the author's supervision, to assist with exposition, editing, and organization of the manuscript.

\begingroup
\small
\bibliographystyle{alpha}
\bibliography{biblio} 
\endgroup

\end{document}